\documentclass[reqno,10pt]{amsart}
\usepackage{amsmath,amssymb,graphics,epsfig,color,cite}
\newtheorem{theorem}{Theorem}[section]
\newtheorem{proposition}[theorem]{Proposition}
\newtheorem{lemma}[theorem]{Lemma}

\theoremstyle{remark}
\newtheorem{remark}[theorem]{Remark}
\theoremstyle{definition}
\newtheorem{definition}[theorem]{Definition}
\numberwithin{equation}{section}
\numberwithin{theorem}{section}
\newcommand{\mc}[1]{{\mathcal #1}}
\newcommand{\bb}[1]{{\mathbb #1}}
\renewcommand{\epsilon}{\varepsilon}   
\newcommand{\sgn}{\mathop{\rm sgn}\nolimits}
\newcommand{\eps}{\epsilon}
\newcommand{\rme}{\mathrm{e}}
\newcommand{\rmd}{\mathrm{d}}
\newcommand{\bam}{\,\overline {\!m\!}\,}

\begin{document}

\title[Boundary effects in the gradient theory of phase
transitions]{Boundary effects in the \\ gradient theory of phase
  transitions}

\author [L.\ Bertini] {Lorenzo Bertini} 
\address{Lorenzo Bertini, Dipartimento di Matematica, SAPIENZA
  Universit\`a di Roma, P.le Aldo Moro 5, 00185 Roma, Italy}
\email{bertini@mat.uniroma1.it}

\author [P.\ Butt\`a] {Paolo Butt\`a}
\address{Paolo Butt\`a, Dipartimento di Matematica, SAPIENZA
  Universit\`a di Roma, P.le Aldo Moro 5, 00185 Roma, Italy}
\email{butta@mat.uniroma1.it}

\author [A.\ Garroni] {Adriana Garroni}
\address{Adriana Garroni, Dipartimento di Matematica, SAPIENZA
  Universit\`a di Roma, P.le Aldo Moro 5, 00185 Roma, Italy}
\email{garroni@mat.uniroma1.it}

\subjclass[2000]{
Primary
82B26,  
49J45;  
Secondary
35J20,  
82B24.  
}

\keywords{Gradient theory of phase transitions, Development by $\Gamma$-convergence, Boundary layer.} 

\begin{abstract}
We consider the van der Waals' free energy functional, with scaling parameter $\eps$, in the plane domain $\bb R_+\times \bb R_+$, with inhomogeneous Dirichlet boundary conditions. We impose the two stable phases on the horizontal boundaries $\bb R_+ \times\{0\}$ and $\bb R_+\times\{\infty\}$, and free boundary conditions on $\{\infty\}\times\bb R_+$. Finally, the datum on $\{0\}\times \bb R_+$ is chosen in such a way that the interface between the pure phases is pinned at some point $(0,y)$. We show that there exists a critical scaling, $y=y_\eps$, such that, as $\eps\to 0$, the competing effects of repulsion from the boundary and penalization of gradients play a role in determining the optimal shape of the (properly rescaled) interface. This result is achieved by means of an asymptotic development of the free energy functional. As a consequence, such analysis is not restricted to minimizers but also encodes the asymptotic probability of fluctuations.
\end{abstract}

\maketitle
\thispagestyle{empty}

\section{Introduction}
\label{sec:1}

The van der Waals' theory of phase transitions \cite{ch,vdw} is based on
the functional
\begin{equation}
  \label{acf}
  E(u) = \int\Big[ \big| \nabla u\big|^2 +  V(u) \Big] \,\rmd r\,,
\end{equation}
where the scalar field $u=u(r)$, $r\in\bb R^d$, represents the local
order parameter and $V(u)$ is a smooth, symmetric, double well
potential whose minimum value, chosen to be zero, is attained at
$u_\pm$; we also assume $V''(u_\pm)>0$.  By introducing a scaling
parameter $\eps>0$, which is interpreted as the ratio between the
microscopic and the macroscopic scale, a most relevant issue is the
asymptotic behaviour of the sequence of functionals
\begin{equation}
  \label{acfeps}
   E_\eps(u) = \int\Big[ \eps\big| \nabla u \big|^2
   +  \frac1\eps V(u) \Big] \,\rmd r\,,
\end{equation}
in the sharp interface limit $\eps\to 0$.  
This has been first analyzed in \cite{mm} and extensively studied
afterwards, see \cite{al} for a review. The limiting functional turns
out to be finite only if $u$ is a function of bounded variation taking
values in $\{u_-,u_+\}$. For $u$ in this set, the limiting functional
is furthermore given by $C_V\, {\mc H}^{d-1}(\mc S_u)$, where $\mc S_u$ denotes the
jump set of $u$ and ${\mc H}^{d-1}(\mc S_u)$ is its $(d-1)$-dimensional Hausdorff measure. The surface energy density $C_V>0$ is finally given by
\begin{equation}\label{cv}
  C_V=\int_{u_-}^{u_+}2\sqrt{V(a)}\,\rmd a\,.
\end{equation}
For any given limiting configuration an optimal sequence can be
constructed by making the transition from the value $u_-$ to the value
$u_+$ in the direction $\nu$ orthogonal to the interface with a one
dimensional profile $\bam(\frac{r\cdot \nu}{\eps})$. Here
$\bam$, the so-called instanton, is the minimizer of the
one-dimensional van der Waals' energy \eqref{acf} with
boundary conditions $u_\pm$ at $\pm\infty$, satisfying $\bam(0)=0$.

As proven in \cite{modica}, when $E_\eps$ is considered together with
Dirichlet boundary conditions, the latter contribute to the limiting
functional with a term taking into account the discrepancy between the
pure phases $\{u_-,u_+\}$ in the interior of the domain and the
prescribed boundary data. We can regard this term as the cost
associated to an interface localized at the boundary. In particular,
when the boundary data take values in the pure phases $\{u_-,u_+\}$,
this cost coincides with the one in the bulk.

Consider a geometry in which the minimizer of the limiting functional
is obtained with an interface localized at the boundary. Of course,
when $\eps$ is small but strictly positive, the minimizer of $E_\eps$
is smooth and the transition between the pure phases takes place in a
thin layer close to the boundary. The purpose of the present paper is
a detailed description, in the two dimensional case, of such boundary
effect by means of an asymptotic development of $E_\eps$. In
particular, such analysis is not restricted to minimizers but also
encodes the asymptotic probability of fluctuations.

We consider the following geometry, see Figure~\ref{fig}. 
As basic domain we choose $\Omega^R=(0,R)\times(0,+\infty)$, $R>0$,
and denote by $x$ and $y$ the horizontal and the vertical coordinates,
respectively. We impose the phase $u_-$ on $(0,R)\times\{0\}$, the
phase $u_+$ on $(0,R)\times\{+\infty\}$, and free boundary conditions
on $\{R\}\times(0,+\infty)$.  Finally, the trace on $\{0\}\times
(0,+\infty)$ is given by a suitable (monotone) continuous function
$v_\eps\colon [0,+\infty)\to [u_-,u_+]$ satisfying $v_\eps(0)=u_-$,
$v_\eps(+\infty)=u_+$.

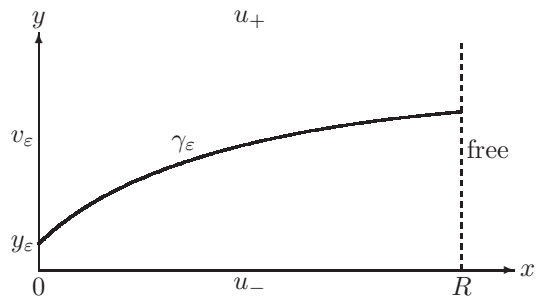
\begin{figure}[h]
\label{fig}
\begin{picture}(200,120)(0,0)
{
%
\thinlines
\put(0,0){\vector(1,0){180}}
\put(185,0){\makebox(0,0){$x$}}
\put(0,0){\vector(0,1){90}}
\put(0,95){\makebox(0,0){$y$}}
\put(0,-6){\makebox(0,0){$0$}}
\multiput(160,0)(0,4){22}{\line(0,1){2}}
\put(160,-6){\makebox(0,0){$R$}}
%
\put(80,-6){\makebox(0,0){$u_-$}}
\put(80,95){\makebox(0,0){$u_+$}}
\put(-6,50){\makebox(0,0){$v_\epsilon$}}
\put(170,46){\makebox(0,0){free}}
%
\thicklines
\qbezier(0,10)(40,50)(160,60)
\put(-6,10){\makebox(0,0){$y_\epsilon$}}
\put(55,48){\makebox(0,0){$\gamma_\epsilon$}}
}
\end{picture}
\caption{
  The domain $\Omega^R$ and the corresponding boundary conditions. 
  The zero of $v_\epsilon$ is $y_\epsilon$ and $\gamma_\epsilon$
  is the interface.
}
\end{figure}

We denote by $E_\eps(\cdot,\Omega^R)$ the functional in \eqref{acfeps}
on the domain $\Omega^R$ with these boundary conditions and let
$u^*_\eps$ be a minimizer of $E_\eps(\cdot,\Omega^R)$. Assume that
$v_\eps$ has a unique zero at $y_\eps$ and let $\gamma_\eps$ be the zero level set of $u^*_\eps$. Observe that $\gamma_\eps$ is a subset of the closure of $\Omega^R$. We shall refer to it as the \emph{interface} and in the following
heuristic discussion we assume that it is the graph of a function on
$[0,R]$, still denoted by $\gamma_\eps$. The boundary condition on
$\{0\}\times(0,+\infty)$ pins the interface at the point $(0,y_\eps)$,
i.e., $\gamma_\eps(0)= y_\eps$. 
We assume that $y_\eps$ converges to zero as $\eps\to 0$. The
result in \cite{modica} then implies that the interface approaches the
interval $[0,R]\times\{0\}$ in the limit $\eps\to 0$,
i.e., $\gamma_\eps\to 0$.
Our aim is a detailed analysis of this convergence, which includes the correction for
finite $\eps$ due to the boundary condition.
There are two competing effects. The boundary datum $u_-$ on
$(0,R)\times \{0\}$ effectively repels $\gamma_\eps$; indeed, in order
to minimize the energy along the one dimensional sections
$\{x\}\times(0,+\infty)$, $x\in (0,R)$, the zero of
$u^*_\eps(x,\cdot)$ should be as large as possible. On the other hand,
the convergence of $\gamma_\eps$ to the flat interface penalizes the
gradient of $\gamma_\eps$. We show that there exists a critical
scaling for $y_\eps$ such that $\gamma_\eps$, properly rescaled,
converges to a non trivial profile for which both effects play a role.

In the spirit of the so-called developments by $\Gamma$-convergence
\cite{ab,marito1,marito2}, we introduce the excess energy  
\begin{equation} 
  \widetilde{E}_\eps(u,\Omega^R)=K_\eps\big[E_\eps(u,\Omega^R)-C_VR\big]\,,
\end{equation}
and look for a sequence $K_\eps\to+\infty$ for which $\widetilde{E}_\eps$
has a non trivial limit. In order to complete this program, we need
however to properly rescale the variables.
The identification of the correct scaling is based on the following 
\emph{ansatz}, which is suggested by the construction of the optimal
sequence in the sharp interface limit of $E_\epsilon$. 
The interface $\gamma_\epsilon$ satifies $\gamma_\eps(x)\approx
y_\eps+\eps\,\phi(x)$ and on each vertical section the function
$u_\eps^*(x,\cdot)$ minimizes the corresponding energy with the
constraint $u_\eps^*(x,\gamma_\eps(x))=0$.  
We thus perform the change of variable $y\mapsto (y-y_\eps)/\eps$
getting
\begin{equation}\label{riscalato}
F_\eps(u)=
K_\eps\bigg\{ \eps^2\int_{0}^{R}\! 
\int_{- \tfrac{y_\eps}{\eps}}^{+\infty}
( \partial_x u)^2 \rmd y\, \rmd  x
  +  \int_{0}^{R}  \!\bigg[  
  \int_{-\tfrac{y_\eps}{\eps}}^{+\infty} 
\Big( ( \partial_y u)^2 +  V(u) \Big) \rmd  y - C_V \bigg] \rmd  x
\bigg\} \,.
\end{equation}
The above expression suggests that in order to appreciate the
variations in the horizontal direction we have to choose $K_\eps=\eps^{-2}$.
Moreover, the analysis of the one-dimensional case in \cite{bbbeq}
implies that the second term on the right hand side of
\eqref{riscalato} is of the order $\exp\{-\beta_V\eps^{-1}y_\eps\}$,
where $\beta_V=\sqrt{2V''(u_+)}$. 
We therefore conclude that the critical scaling for $y_\eps$ is given
by $y_\eps\approx\frac{2}{\beta_V}\eps\log\eps^{-1}$.

In this paper we analyze the variational convergence of the
functionals $F_\epsilon$ in \eqref{riscalato}; referring to the next
section for the precise statement, we here discuss informally our
results.  
If $y_\eps\gg \frac{2}{\beta_V} \eps\log\eps^{-1}$ we show that the
repulsion due to the boundary is not seen in the limiting functional
and that the rescaled profile corresponding to the minimizer $u^*_\eps$ is
flat, i.e., $\phi=0$.
In the critical scaling $y_\eps= \frac{2}{\beta_V}\eps\log\eps^{-1}$ we prove
that the $\Gamma$-limit of the functionals $F_\epsilon$ in
\eqref{riscalato} is finite only on functions $u$ of the form
$u(x,y)=\bam(y-\phi(x))$, with $\phi(0)=0$. On these functions the
limiting functional is furthermore given  by
\begin{equation}
  \label{limit}
  \int_{0}^{R} \Big[ \tfrac{C_V}{2} \, \phi'(x)^2 
  + B_V  \,\rme^{-\beta_V \phi(x)} \Big] \,\rmd x \,,
\end{equation}
for a suitable constant $B_V>0$ that can be computed explicitly.  
As a consequence of this $\Gamma$-convergence result we deduce the
sharp asymptotic for the minimizer $u_\eps^*$ of the original
functional $E_\epsilon$ in \eqref{acfeps}. Namely,
\begin{equation}
  \label{approx}
  u_\eps^*(x,y)\approx \bam\big(\tfrac{y-\gamma_\eps(x)}{\epsilon} \big)\,,
  \qquad \gamma_\eps(x)\approx y_\eps+\eps\,\phi^*(x)\,,
\end{equation}
where $\phi^*$ is the minimizer of the energy \eqref{limit} with the
boundary condition $\phi(0)=0$. A simple computation then shows,
\begin{equation}
  \label{phi*}
  \phi^*(x)= \frac{2}{\beta_V}
  \log\Big(1 + {\beta_V}\sqrt{\tfrac{B_V}{2C_V}} \, x \Big)\,.
\end{equation}
In the case $y_\epsilon\ll \frac{2}{\beta_V} \eps\log\eps^{-1}$, when $x$ is close to zero, the repulsion from the boundary is much stronger than the penalization on the gradients of $\gamma_\epsilon$. This implies that for each fixed $x$ close to zero we have $\gamma_\epsilon(x)- y_\epsilon\gg \epsilon$. On the other hand, if $x$ is such that $\gamma_\epsilon(x)\approx \frac{2}{\beta_V} \eps\log\eps^{-1}$ we are back in the situation described by the critical scaling.  We therefore expect, but do not prove here, that in this regime the asymptotic expression of the interface $\gamma_\epsilon$ for $x$ bounded away from $0$ has still the form  $\gamma_\eps(x)\approx  \frac{2}{\beta_V} \eps\log\eps^{-1} +\eps\,\phi^*(x)$. 

We conclude with few remarks on the relationship of the problem here
considered with some (microscopic) statistical mechanics models.  
In the context of short-range, Ising-like models, the statistical
properties of an interface above a wall have been mostly studied
for the so-called \emph{effective interface models}, see \cite{F} for
a review. These models are obtained by assuming that the interface can be
described as the graph of some function $\phi\colon\Lambda\to\bb R_+$,
where $\Lambda$ is a finite subset of the lattice $\bb Z^{d-1}$.  
One then introduces a Gibbs measure on the set of the interface
configurations, with a short range energy term penalizing the
gradients of $\phi$, and analyzes the asymptotic behaviour of this
measure as $\Lambda$ invades $\bb Z^{d-1}$.  While the energy is
minimized by an interface localized at the wall, i.e., $\phi=0$, the
presence of the fluctuations induces a repulsion, e.g., the expected
value of $\phi$ diverges as $\Lambda\uparrow\bb Z^{d-1}$. This effect
is referred to as \emph{entropic repulsion}: for the interface it
is more convenient to have some room to fluctuate rather than to
minimize the energy.

The asymptotic \eqref{approx} does not reflect an entropic repulsion effect. In the case of the van der Waals' functional, the repulsion from the wall is in fact due to an energetic effect induced by the boundary conditions. The case of long-range, Kac-like models is, on the other hand, much closer to the
problem considered here. Indeed, on a suitable mesoscopic scale the
behavior of those models is well described by a free energy functional
which, although non local, has similar features to \eqref{acfeps}, see
\cite{Err}. In particular, the corresponding sharp interface limit has
been analyzed in \cite{abcp}, where it is shown that the $\Gamma$-limit 
of the free energy functional is proportional to the perimeter 
of the interface between the pure phases, the proportionality constant
identifying the surface tension. As far as we know, the asymptotic 
behaviour of an interface close to a wall has not been analyzed in detail 
for systems with Kac-type interactions, but it seems reasonable that the 
effects here discussed are also relevant in such a situation.

\section{The main result} 
\label{sec:2}

For the sake of concreteness, we restrict the analysis to the paradigmatic case of the symmetric double well potential, i.e., we choose
    \begin{equation}
    \label{p1}
V(u) =  \big(u^2-1\big)^2,
    \end{equation}
which attains its minimum at $u_\pm=\pm 1$. With this choice, the instanton $\bam$ is given by $\bam(y)=\tanh y$ and elementary computations show $C_V=\frac{8}{3}$, $\beta_V=4$, $B_V=16$ \cite[Appendix A]{bbbeq}.

As reference domain we choose the quadrant of $\bb R^2$ given by $\Omega_\ell=(0,+\infty)\times (-\ell,+\infty)$, $\ell>0$. The case of the half strip $(0,R)\times (-\ell,+\infty)$ discussed in the Introduction can be analyzed by the same arguments. The parameter $\ell$ has been introduced in such a way that the zero of the trace on $\{0\}\times (-\ell,+\infty)$ approaches zero as $\eps\to 0$. Accordingly, the asymptotic expansion of the functional will be discussed in the fine tuning $\ell = \frac 12 \log\eps^{-1} + O(1)$, which corresponds to the critical scaling discussed in the Introduction.

Given $m\colon(-\ell,+\infty)\to \bb R$ we introduce the one dimensional functional
\begin{equation}
  \label{1-d}
\mc F_\ell(m):= 
  \int_{-\ell}^{+\infty}
  \Big[  (m')^2 + V(m) \Big]\,\rmd y\,.
   \end{equation}
With this notation, we can rewrite the functional in \eqref{riscalato}  as
\begin{equation}
  \label{3.4bis}
  F_\epsilon(u):=\int_{0}^{+\infty} \int_{-\ell_\epsilon}^{+\infty}
  ( \partial_x u)^2\, \rmd y \,\rmd x
  +  \epsilon^{-2} \int_{0}^{+\infty} \Big[
  \mc F_{\ell_\eps}(u(x,\cdot))- \frac 83 \Big] \, \rmd x \,,
\end{equation}
where $\ell_\eps\to+\infty$ as $\eps\to 0$.
We will study the asymptotic behaviour, as $\eps\to 0$, in terms of $\Gamma$-convergence of $F_\eps$  subject to the following boundary conditions,
\begin{equation}\label{bc}
\begin{cases}
u(0,y)=w_\epsilon(y) &\hbox{$y\in[-\ell_\eps,+\infty)\, ,$} \\
u(x,-\ell_\eps)=-1&\hbox{$x\in (0,+\infty)\,,$}\\
u(x,\cdot)-1 \in L^2((0,+\infty)) &\hbox{$x\in(0,+\infty)$\,,}
\end{cases}
\end{equation}
for a suitable continuous function $w_\epsilon\colon[-\ell_\eps,+\infty)\to \bb R$ with $w_\eps(-\ell_\eps) = -1$. We shall regard $u$ as a function on $[0,+\infty)\times\bb R$ by setting $u=-1$ on $[0,+\infty)\times (-\infty,-\ell_\eps)$. Accordingly, $w_\eps$ is also regarded as a function in $\bb R$ by setting $w_\eps = -1$ on $(-\infty,-\ell_\eps)$. 

We set $\chi(x,y)= \chi(y):=\sgn(y)$ and define the affine space
\begin{equation*}
X:=\{u\,: u-\chi\in L^2((0,R)\times \bb R)\ \hbox{for any } R>0\}\,,
\end{equation*}
endowed with the metric
\begin{equation*}
d_X(u,v)=\sum_n 2^{-n}\left( 1\wedge \|u-v\|_{L^2((0,n)\times \bb R)}\right)\,. 
\end{equation*}
We shall then regard $F_\epsilon$ as a functional on $X$, which takes value $+\infty$ whenever $u$ does not satisfy the boundary conditions \eqref{bc} or $u$ is not identically equal to $-1$ on $[0,+\infty)\times (-\infty,-\ell_\eps)$. Note that $F_\eps(u)<+\infty$ implies $u\in H^1((0,R)\times\bb R)$ for any $R>0$, and therefore the condition $u(0,\cdot)=w_\eps$ can be understood in terms of traces.

It turns out that the $\Gamma$-limit of $F_\eps$ depends on the limit of $\ell_\epsilon- \frac12\log\epsilon^{-1}$ that will be denoted by $\alpha$. We first introduce the limiting functional. Given $\alpha\in\bb R$, we let $\mc G^\alpha:C([0,+\infty))\to [0,+\infty]$ be the lower semicontinuous, with respect to the uniform convergence on compacts, functional defined by
\begin{equation}
  \label{3.5}
  \mc G^\alpha (\phi) = \int_{0}^{+\infty} \Big[ \frac43 \phi'(x)^2 
  + 16 \, e^{-4\alpha} \,\rme^{-4 \phi(x)} \Big] \,\rmd x \,.
\end{equation}
Recall that $\bam$ is the minimizer of the one-dimensional van der Waals' energy \eqref{acf} with boundary conditions $\pm 1$ at $\pm\infty$ satisfying  $\bam(0)=0$, and denote by $\bam_z$ its translated by $z\in \bb R$, i.e., $\bam_z(y):=\bam(y-z)$ for $y\in \bb R$. We then let $F^{\alpha}\colon X \to [0,+\infty]$ be defined by
\begin{equation}
  \label{3.6}
 F^{\alpha}(u) :=
  \begin{cases}
    \mc G^\alpha (\phi) & \textrm{if } u = \bam_\phi  
    \textrm{ for some $\phi\in C([0,+\infty))$, }  
    \phi(0)=0\,,\\ 
    + \infty & \textrm{otherwise},
  \end{cases}
\end{equation}
where we understand $\bam_\phi(x,y)=\bam_{\phi(x)}(y)$.
\begin{theorem}
  \label{teorema1} 
  Assume $\lim_{\epsilon\to 0}\big[ \ell_\epsilon-
  \frac12\log\epsilon^{-1}]= \alpha$, 
  \begin{equation}
  \label{datobordo}
  \lim_{\eps\to 0} \eps^{-1}\Big(\mc F_{\ell_\eps}(w_\eps)-\frac83\Big)
  = 0\, , \quad \hbox{and}\quad \lim_{\eps\to 0} w_\eps(0)=0\,. 
\end{equation}
Then 
\par\noindent \emph{(Compactness)} If a sequence $u_\eps$ satisfies $\limsup_\eps F_\epsilon(u_\epsilon)< +\infty$, then for any
  $R>0$,
  \begin{equation}\label{tilde0}
\lim_{\epsilon\to 0}\big\|u_\eps-\bam_{\phi_\eps}\big\|_{L^2((0,R)\times\bb R)}^2  = 0\,,
    \end{equation}
for some  sequence $\phi_\epsilon$ precompact in $C([0,+\infty))$, satisfying $\phi_\epsilon(0)\to 0$ as $\eps\to 0$.
\par\noindent In particular, the sequence $F_\epsilon$ is equicoercive in $X$. 
\par\noindent \emph{($\Gamma$-convergence)} The sequence $F_\epsilon$
    $\Gamma$-converges to $F^\alpha$ as $\epsilon\to 0$, i.e.,~for any
    $\phi\in C([0,+\infty))$, with $\phi(0)=0$ we have,
        \begin{itemize}   
        \item[(i)] \emph{($\Gamma$-liminf)} if
          $u_\epsilon\to\bam_\phi$ in $X$, then
        \begin{equation*}
        \liminf_{\epsilon\to 0} F_\epsilon(u_\epsilon)\geq \mc G^\alpha(\phi)\,;
        \end{equation*}
      \item[(ii)] \emph{($\Gamma$-limsup)} there exists
        $u_\epsilon\to\bam_\phi$ in $X$ such that
        \begin{equation*}
        \lim_{\epsilon\to 0} F_\epsilon(u_\epsilon)=\mc G^\alpha(\phi)\,.
        \end{equation*}
  \end{itemize}
\end{theorem}
  \begin{remark}
  \label{teorema1bis} The above result holds also in the case $\alpha=+\infty$. More precisely, if $\ell_\eps$ satisfies  $\lim_{\epsilon\to 0}\big[ 
  \ell_\epsilon- \frac12\log\epsilon^{-1}]=+\infty$ and $w_\eps$ satisfies \eqref{datobordo},
  the statements in Theorem~\ref{teorema1} hold true with $\mc G^\alpha(\phi)$
  replaced by $ \frac 43 \int_0^{+\infty} |\phi'|^2\rmd x$.
\end{remark}
 
By standard properties of $\Gamma$-convergence, see e.g.~\cite[Theorem~1.21]{marito1}, the above results imply the convergence of the
minimizers of $F_\eps$ to the (unique) minimizer of the corresponding
limiting functionals. In particular, in the critical scaling, the
repulsion of the boundary conditions on $\partial \Omega_{\ell_\eps}$ 
competes with the tendency of being flat and the minimizers of $F_\eps$ converge in $X$ to $\bam_{\phi_\alpha^*}$, where $\phi^*_\alpha(x)=\frac{1}{2}\log\big(1 +\rme^{-2\alpha} \, 4\sqrt3\,x\big)$ is the minimizer of $\mc G^\alpha$ with the boundary condition $\phi(0)=0$.
On the other hand, when $\ell_\eps\gg \frac{1}{2}\log\eps^{-1}$, the
repulsion of the boundary conditions on $\partial \Omega_{\ell_\eps}$ is 
not felt and the optimal interface is flat. This is consistent with the 
fact that as $\alpha\to +\infty$ we have $\phi^*_\alpha(x)\to 0$, $x\ge 
0$. When $\alpha\to -\infty$ we instead have $\phi^*_\alpha (x) + \alpha
\to \frac 12 \log \big( 4\sqrt{3} \, x\big)$, $x>0$. This corresponds to
the situation, as described in the Introduction, in which $y_\epsilon\ll
\frac{2}{\beta_V} \eps\log \eps^{-1}$.  We emphasize that
the latter convergence of $\phi^*_\alpha$ cannot be described in 
variational terms because the amount of energy, as measured by $\mc G^\alpha$, stored in any neighborhood of zero diverges as $\alpha\to -\infty$ while the energy stored in any neighborhood not containing zero remains finite and strictly positive. 

\smallskip
The rest of the paper is organized in the following way. Section~\ref{sec:3} and Appendix~\ref{sec:a} are devoted to a detailed study of the asymptotic expansion by $\Gamma$-convergence of the one dimensional functional $\mc F_\ell$ in \eqref{1-d} as $\ell\to+\infty$. Such analysis is a preliminary tool for the proof of Theorem~\ref{teorema1}, which is the content of  Section~\ref{sec:4} (compactness) and Section~\ref{sec:5} ($\Gamma$-convergence).

\section{One-dimensional problem}
\label{sec:3}

In this section we analyze the one dimensional functional $\mc F_\ell$ defined in \eqref{1-d}. The development by $\Gamma$-convergence of $\mc F_\ell$ as $\ell\to+\infty$ is studied in \cite{bbbeq}. Here we prove  quantitative estimates related to that asymptotic expansion. Hereafter we shorthand $L^2(\bb R)$ and $H^1(\bb R)$ by $L^2$ and $H^1$, respectively.

Recalling $\chi(y) = \sgn(y)$, we set $\mc X:=\{m\colon m-\chi\in L^2\}$, that we consider endowed with the strong $L^2$-topology. Given $\ell>0$, we let $\mc X_\ell\subset\mc X$ be the closed subspace defined by
\begin{equation}
\label{bc1}
\mc X_\ell =\{ m\in \mc X\,:\ m(y)=-1\;\hbox{ if $y\in(-\infty,-\ell)$}\}\,.
\end{equation}
We then regard $\mc F_\ell$ as a functional on $\mc X$ which takes value $+\infty$ whenever $m\notin\mc X_\ell$. It is simple to show that the sequence of functionals $\mc F_\ell$ $\Gamma$-converges to the functional $\mc F\colon \mc X \to [0,+\infty]$, defined by 
\begin{equation}
\label{mcf}
\mc F(m) := \int_{-\infty}^{+\infty} \Big[  (m')^2 + V(m) \Big]\,\rmd y\,.
\end{equation}
By the well known Modica-Mortola trick \cite{mm}, 
\begin{equation}
\label{unimc}
\min\mc F= C_V =\frac 83, \qquad \mathrm{arg\,min}\,\mc F = \{\bam_z,\, z\in\bb R\}\,.
\end{equation}

Given $z\in(-\ell,+\infty)$ we define
\begin{equation}
\label{nome}
m^\ell_z=\mathrm{arg\,min}\,\big\{\mc F_\ell(m)\colon\ m\in\mc X_\ell ,\ m(z)=0\big\}\, ,
\end{equation}
observing that the minimizer is unique. We introduce the one dimensional manifold $\mc M^\ell:=\{m_z^\ell\colon z\in(-\ell,+\infty)\}$ in $\mc X$. Sometimes, we will use the notation $m^\ell_z(\cdot)= m^\ell(\cdot,z)$. If $y>z$ then $m^\ell_z(y)=\bam_z(y)$. Moreover, for $y\in (-\ell,z)$, $m^\ell_z$ coincides with the (unique) solution to the following boundary value problem,
\begin{equation}
\label{eulero}
\begin{cases}
-2m'' +V'(m)=0\ & \hbox{in}\ (-\ell,z)\, ,\\
m(-\ell)=-1\,,\ m(z)=0\,.
\end{cases}
\end{equation}
We next state, referring to the Appendix~\ref{sec:a} for the proof, sharp estimates concerning $m_z^\ell$ and its convergence to $\bam_z$. 
\begin{proposition}
\label{prop_ml}
There exists a constant $A$ such that, for any $\ell>0$ and $z\in \bb R$ satisfying $\ell+z\geq 1$,
\begin{eqnarray}
\label{stim1}
&& \sup_{y\in(-\ell,z)} |m^\ell_z(y)-\bam_z(y)| \le  A\rme^{-2(\ell+z)}, \\
\label{mellezeta}
&&\sup_{y\in(-\ell,z)} |\partial_z m^\ell_z(y)+\bam_z'(y)| \le  A\rme^{-2(\ell+z)},  \\ \label{mellezeta2}
&& \sup_{y\in(-\ell,z)} |\partial_{zz} m^\ell_z(y)-\bam_z''(y)| \le  A\rme^{-2(\ell+z)}\,, 
\\ \label{salto}
&&  \left[(m^\ell_z)'\right](z) +  \left[\partial_z m^\ell_z\right](z) = 0, \qquad \left|\left[(m^\ell_z)'\right](z)\right|\leq A\rme^{-4(\ell+z)}\,,
\end{eqnarray}
where $[f](z)$ denotes the jump of the function $f$ at $z$. Moreover, for any $\ell$, $z_1$, and $z_2$ such that $(z_1+\ell)\wedge(z_2+ \ell)\geq 1$,
\begin{equation}
\label{lip}
\frac1{A}(|z_1-z_2|^2\wedge |z_1-z_2|)\leq \|m^\ell_{z_1}-m^\ell_{z_2}\|^2_{L^2}\leq A (|z_1-z_2|^2\wedge |z_1-z_2|) \, .
\end{equation}
\end{proposition}
\begin{remark}
\label{35}
Since $m^\ell_z(y)=\bam_z(y)$ for $y>z$ and $m^\ell_z(y)=-1$ for $y<-\ell$ the above bounds and \eqref{eulero} yield that $m^\ell_z$ converges to $\bam_z$ in $H^2$. In particular,
\begin{equation}
\label{dz}
\lim_{\ell\to +\infty}\int_{-\ell}^{+\infty}\! (\partial_z m^\ell_z)^2\, \rmd y=\int_{-\infty}^{+\infty}\! \bam_z'(y)^2 \, \rmd y = \frac43\,,
\end{equation}
uniformly with respect to $z\in [\bar z_\ell, +\infty)$ with 
$\ell+\bar z_\ell\to +\infty$.
\end{remark}

The following lemma is proven in \cite[Lemma~A.1]{bbbeq}. It is the key ingredient to study the development by $\Gamma$-convergence of the functionals $\mc F_\ell$.
\begin{lemma}
\label{fucking}
Let $z\in\bb R$ and $z_\ell$ be a sequence  converging to $z$. Then
\begin{equation*}
\lim_{\ell\to+\infty} \rme^{4\ell}\Big[\mc F_\ell(m_{z_\ell}^\ell)-\frac83\Big]=16\,\rme^{-4z}
\end{equation*}
and, given $\bar z\in\bb R$, this limit is uniform for $z\in [\bar z,+\infty)$.
\end{lemma}
The notion of center for  functions in $\mc X_\ell$, introduced in \cite{bdp}, will play an important role in our analysis. 
\begin{definition}
\label{centro}
Given $m\in \mc X_\ell$ we say that $\zeta\in(-\ell,+\infty)$ is a \emph{center} of $m$ if
\begin{equation*}
\zeta\in \mathrm{arg\,min}\left\{\|m-m^\ell_z\|^2_{L^2}\colon z\in(-\ell,+\infty)\right\}\,.
\end{equation*}  
In particular, the function $m^\ell_\zeta$ is a $L^2$-projection of $m$ on the manifold $\mc M^\ell$.
\end{definition}
Referring to \cite{bbb} for a dynamical interpretation of the above definition, we simply note that if $\zeta$ is a center of $m$ then the following orthogonality condition holds,
\begin{equation}
  \label{ort}
  \int_{-\ell}^{+\infty}\!
  \big[ m(y)-m^\ell_\zeta(y) \big] \, \partial_z m^\ell_\zeta(y) \,\rmd y=0\, ,
\end{equation}
where $\partial_z m^\ell_\zeta(y)=\partial_z m^\ell_z(y){\big|_{z=\zeta}}$.
We next introduce a suitable neighborhood of the manifold $\mc M^\ell$, which takes into account the boundary conditions \eqref{bc1}. More precisely, given $\delta>0$ and $k>0$ we set
\begin{equation}
\label{palma}
\mc T^\ell(\delta,k):=\left\{m\in \mc X_\ell\colon \exists\, z\in(-\ell+k,+\infty) \hbox{ such that } \|m-m^\ell_z\|_{H^1}<\delta  \right\}\, .
\end{equation}
The following result shows, in particular, that if $m$ is such that $\mc F_\ell(m)$ is close to its minimum then $m$ is close to the manifold $\mc M^\ell$.
\begin{theorem}
\label{lemma0} 
The following statements hold.
\begin{itemize}
\item[(i)] For each $\delta>0$ and $\kappa>0$ there exist $\eta>0$ and $\ell_0>0$ such that if $\mc F_\ell(m)-\frac83<\eta$ for some $\ell\geq\ell_0$ then $m\in\mc T^\ell(\delta,\kappa)$.  
\item[(ii)] There exist constants $\ell_1$, $\delta_1$, $\kappa_1$, and $C_1$ such that, for all $\ell\geq \ell_1$, $\delta\leq \delta_1$, and $\kappa\geq\kappa_1$, if $m\in\mc T^\ell(\delta,\kappa)$ then the center $\zeta$ of $m$ is unique, satisfies $\zeta>-\ell+\kappa-2\delta$, and
\begin{equation}
\label{gapimprove}
\|m-m^\ell_\zeta\|^2_{H^1}\leq C_1\left[\mc F_\ell(m)-\mc F_\ell(m^\ell_\zeta) +\rme^{-4(\ell+\zeta)}\|m-m^\ell_\zeta\|_{H^1}\right]\,.
\end{equation}
\item[(iii)] For each $\bar z\in\bb R$ there exist two positive constants $C_2$ and $\ell_2$ such that, for any $\ell>\ell_2$ and $z\in [\bar z,+\infty)$, 
\begin{equation*}
\phantom{merdaa} \mc F_\ell(m)-\mc F_\ell(m^\ell_z) \leq C_2 \left(\|m-m^\ell_z\|^2_{H^1}+ \|m-m^\ell_z\|^4_{H^1}+\rme^{-4\ell}\|m-m^\ell_z\|_{H^1}\right)\,,
\end{equation*}
for all $m\in \mc X_\ell$.
\end{itemize}
\end{theorem}
We emphasize that while in statement (ii) of the above theorem $\zeta$ denotes the center of $m$, in statement (iii) $z$ is arbitrary.
\begin{remark}
\label{remgap}
As a consequence of \eqref{gapimprove} and Lemma~\ref{fucking}, there exist constants $\ell_0$, $\delta$, $\kappa$, and $C_0$ such that, for all $\ell\geq \ell_0$, if $m\in\mc T^\ell(\delta,\kappa)$, then the center $\zeta$ of $m$ is unique and 
\begin{equation}
\label{gap}
\|m-m^\ell_\zeta\|^2_{H^1}+\rme^{-4(\ell+\zeta)}\leq C_0\left[\mc F_\ell(m)-\frac83\right]\,.
\end{equation}
\end{remark}
\begin{proof}[Proof of Theorem~\ref{lemma0}] The proof is split into separate arguments. In the sequel we denote by $C$ a strictly positive constant, independent of $\ell$ and $\zeta$, whose numerical value may change from line to line. 

\smallskip\noindent
{\it Proof of statement (i), step 1.} Here we prove that for each $\delta>0$ there exist $\eta>0$ and $\ell_0>0$ such that if $\mc F_\ell(m)-\frac83<\eta$ for some $\ell\geq\ell_0$, then ${\rm dist}_{H^1}(m,{\mc M}^\ell)<\delta$. We argue by contradiction and assume that there exist $\delta_0>0$ and a sequence $m_\ell$ with 
\begin{equation}
\label{contra0}
\liminf_{\ell\to +\infty}\,{\rm dist}_{H^1}(m_\ell,{\mc M}^\ell)\geq\delta_0
\end{equation} 
such that
\begin{equation}
\label{contra1}
\limsup_{\ell\to+\infty}\,{\mc F}_\ell(m_\ell)\leq \frac83\,.
\end{equation}
Note that by \eqref{contra1} the function $m_\ell$ satisfies the boundary conditions \eqref{bc1}. We set $z_\ell=\inf\{y\colon m_\ell(y)=0\}$ and define $\widetilde m_\ell(y)=m_\ell(y+z_\ell)$, so that $\widetilde m_\ell(0)=0$. The boundedness of the energy implies that $\widetilde m_\ell$ converges, up to a subsequence, to some continuous function $m_0$, uniformly in compacts. We will show that $m_0=\bam$ and that $\widetilde m_\ell-\bam$ actually  converges to zero in $H^1$.

Given $\sigma>0$ we set
\begin{equation*}
\begin{split}
a_\ell&=\sup\{y<z_\ell\,:\ m_\ell(y)>-1+\sigma\}\,,\\
b_\ell&=\inf\{y>z_\ell\,:\ m_\ell(y)<1-\sigma\}\,.
\end{split}
\end{equation*}
The boundedness of $\int_{-\infty}^{+\infty}\!V(m_\ell)\,\rmd y$ implies that $b_\ell-a_\ell\leq C_\sigma$, for some constant $C_\sigma$ independent of $\ell$. This guarantees that  the energy of $\widetilde m_\ell$ does not escape to infinity. More precisely, using the Modica-Mortola trick,
\begin{equation*}
\begin{split}
\int_{-C_\sigma}^{C_\sigma}\! \left[|{\widetilde m}'_\ell|^2 +V(\widetilde m_\ell)\right]\,\rmd y &\geq2\int_{a_\ell-z_\ell}^{b_\ell-z_\ell}\! |{\widetilde m}'_\ell|\sqrt{V(\widetilde m_\ell)}\,\rmd y \\ &=2\int_{-1+\sigma}^{1-\sigma}\! \sqrt{V(m)}\,\rmd m=\frac83 -  4\sigma^2\left(1-\tfrac\sigma 3\right)\,,
\end{split}
\end{equation*}
we deduce, taking into account \eqref{contra1}, that
\begin{equation}
\label{limi}
\lim_{\ell\to+\infty}{\mc F}_\ell(m_\ell) =\frac83\,
\end{equation}
and thence
\begin{equation}
\label{complement}
\lim_{\sigma\to 0}\; \limsup_{\ell\to+\infty} \int\limits_{[-C_\sigma,C_\sigma]^\complement}\! \left[|{\widetilde m}'_\ell|^2 +V(\widetilde m_\ell)\right] \, \rmd y=0\,.
\end{equation}
Therefore, up to a subsequence, 
\begin{equation*}
\lim_{\ell\to+\infty} \int_{-\infty}^{+\infty}\!  V(\widetilde m_\ell)\,\rmd y = \int_{-\infty}^{+\infty}\!  V(m_0) \,\rmd y\,,
\end{equation*}
so that 
\begin{equation}
\label{secon}
\lim_{\ell\to+\infty} \int_{-\infty}^{+\infty}\!  |\widetilde m_\ell'|^2 \,\rmd y = \int_{-\infty}^{+\infty}\!  | m_0'|^2 \,\rmd y\,.  
\end{equation}
In particular, by \eqref{limi}, $\mc F(m_0)=\frac 83$. Since $m_0(0)=0$, by the uniqueness up to translations of the minimizer of $\mc F$, recall \eqref{unimc}, $m_0=\bam$. Now, using \eqref{complement} and the definition of $a_\ell$ and $b_\ell$, we get the convergence of $\widetilde m_\ell$ to $\bam$ in $L^2$. This, together with \eqref{secon} and Remark~\ref{35}, contradicts \eqref{contra0} and then concludes the proof of the step.

\smallskip\noindent
{\it Proof of statement (i), step 2.} Here we conclude the proof. Again we argue by contradiction and assume that there exist $\delta,\kappa>0$ and a sequence $m_\ell\notin\mc T^\ell(\delta,\kappa)$ such that $\mc F_\ell(m_\ell)\to\frac83$. By step 1 it is enough to consider the case when there exists a sequence $z_\ell\in(-\ell,-\ell+\kappa)$ such that $\|m_\ell-m^\ell_{z_\ell}\|_{H^1}<\delta$. This yields $\|m_\ell-m^\ell_{z_\ell}\|_{L^\infty}<C\delta$, and hence, if $z'_\ell$ is any zero of $m_\ell$, then $|m_{z_\ell}(z_\ell')|<C\delta$. By Proposition~\ref{prop_ml}, this implies $|z_\ell-z_\ell'|<C\delta$. On the other hand, it is easy to see that 
\begin{equation*}
\min\left\{{\mc F}(m)\colon m(a)=-1\, ,\ m(b)=0\right\}> \frac43\qquad \forall \, a, b\colon  -\infty<a<b<+\infty\,.
\end{equation*}
Therefore,
\begin{equation*}
\liminf_{\ell\to +\infty}\mc F_\ell(m_\ell)\geq \min\left\{{\mc F}(m)\colon m(0)=-1,\  m(\kappa+C\delta)=0\right\} +\frac43>\frac83\,.
\end{equation*}
This is a contradiction and concludes the proof of statement (i).

\smallskip\noindent
{\it Proof of statement (ii).} The uniqueness of the center, for $\delta_1$ small enough and $\kappa_1$ large enough, is stated in \cite[Proposition~3.1]{bbb}. The proof follows by standard implicit function argument \cite{bdp}. That proposition also guarantees that the center $\zeta$ satisfies the bound $\zeta>-\ell+\kappa-2\delta$.

Let $m\in\mc T^\ell(\delta,\kappa)$ and  $\zeta\geq-\ell+\kappa-2\delta$ be the unique center of $m$. Recalling that $m(y)-m^\ell_\zeta(y) = 0$ for $y\in (-\infty,-\ell]$, we decompose 
\begin{equation*}
\mc F_\ell(m)=\mc F_\ell(m_\zeta^\ell)+ I^1_\ell+ I^2_\ell + I^3_\ell\, ,
\end{equation*}
where
\begin{equation*}
\begin{split}
&I^1_\ell=\int_{-\infty}^{+\infty}\!  \big[2\partial_y m^{\ell}_\zeta\partial_y(m-m^\ell_\zeta)+ V'(m^{\ell}_\zeta)(m-m^\ell_\zeta)\big]\, \rmd y\,,\\
&I^2_\ell=\int_{-\infty}^{+\infty}\!  \Big[(\partial_y(m-m^\ell_\zeta))^2+\frac12 V''(m^{\ell}_\zeta)(m-m^\ell_\zeta)^2\Big]\, \rmd y\,,\\
&I^3_\ell=\int_{-\infty}^{+\infty}\!  \Big[\frac16 V'''(m^{\ell}_\zeta)(m-m^\ell_\zeta)^3+ \frac1{24} V''''(m^{\ell}_\zeta)(m-m^\ell_\zeta)^4\Big]\, \rmd y\,.
\end{split}
\end{equation*}
The proof will be achieved by analyzing in detail the quadratic form in $I_\ell^2$ and showing that it can be bounded from below by $\|m-m^\ell_\zeta\|^2_{H^1}$, while the other two terms will be bounded in absolute value. 

We first estimate $I_\ell^1$. By integration by parts, using \eqref{eulero} and \eqref{salto} we get
\begin{equation}
\label{i1}
\begin{split}
|I^1_\ell|&=2 \left|(m(\zeta)-m^\ell_\zeta(\zeta))\left[\partial_ym^\ell_\zeta\right](\zeta)\right|\\ &\leq C\rme^{-4(\ell+\zeta)}|m(\zeta)-m^\ell_\zeta(\zeta)|
\leq C \rme^{-4(\ell+\zeta)}\|m-m^\ell_\zeta\|_{H^1}\,,
\end{split}
\end{equation}
where we have used the Sobolev embedding. As for the term $I_\ell^2$, the application of the Sobolev embedding yields,
\begin{equation}
\label{i3}
|I^3_\ell|\leq C\left( \|m-m^\ell_\zeta\|_{H^1}^3+\|m-m^\ell_\zeta\|_{H^1}^4\right)\,.
\end{equation}
Finally, it remains to estimate $I^2_\ell$. We show that
\begin{equation}
\label{i2}
I_\ell^2\geq \frac 1C \|m-m_\zeta^\ell\|_{H^1}^2\,.
\end{equation}
We denote by $\mc H^\ell_\zeta$ the Schr\"odinger operator on $L^2((-\ell,+\infty))$ defined as
\begin{equation*}
\mc H^\ell_\zeta= -\frac{\rmd^2}{\rmd y^2}+ V''(\bam_\zeta)\,,
\end{equation*}
with domain $H^2((-\ell,+\infty))\cap H^1_0((-\ell,+\infty))$. In the sequel, we shall regard to $L^2((-\ell,+\infty))$ as a subset of $L^2$, by setting, for every function $\psi\in L^2((-\ell,+\infty))$,  $\psi(y) = 0$ if $y\in (-\infty,-\ell]$. Let also set $\varphi=\varphi_{\ell,\zeta}=m-m_\zeta^\ell$. With this notation we rewrite $I^2_\ell$ as the quadratic form,
\begin{equation}
\label{quadratic}
I_\ell^2=\langle \varphi, \mc H^\ell_\zeta\varphi\rangle_{L^2} + \langle \varphi, (V''(m^\ell_\zeta)-V''(\bam_\zeta))\varphi\rangle_{L^2}\,,
\end{equation}
where $\langle\cdot,\cdot\rangle_{L^2}$ denotes the $L^2$-inner product. By \eqref{stim1}, the second term on the right hand side of the above equality is bounded in absolute value by $C\rme^{-2(\ell+\zeta)} \|\varphi\|_{L^2}^2$.

It remains to estimate the first term of the right hand side of \eqref{quadratic}. As shown in \cite[Theorem~3.2]{bbb} the first eigenvalue $\lambda_\zeta^\ell>0$ of  the operator  $\mc H^\ell_\zeta$ is exponentially small as  $\ell\to +\infty$ while the remaining part of the spectrum is bounded away from  zero uniformly in $\ell$ and $\zeta$ (since $\ell+\zeta>\kappa_1-2\delta_1$). We denote by $\Psi_\zeta^\ell$ the eigenfunction corresponding to the eigenvalue $\lambda_\zeta^\ell$. From these results it follows that there exists a constant $g_1>0$, independent of $\ell$ and $\zeta$, such that for any  $\psi\in L^2((-\ell,+\infty))$, $\psi\perp\Psi_\zeta^\ell$, i.e., satisfying $\langle \psi,\Psi_\zeta^\ell\rangle_{L^2}=0$,
\begin{equation}
\label{gap1}
\langle \psi, \mc H^\ell_\zeta\psi\rangle_{L^2}\geq g_1 \langle \psi,\psi\rangle_{L^2}\,.
\end{equation}
We next improve the above bound with the $H^1$-norm. More precisely, we prove that there exists a constant $\bar g_1>0$ independent of $\ell$ and $\zeta$, such that
\begin{equation}
\label{gap2}
J_{\ell,\zeta}:= \inf_{\psi\perp \Psi_\zeta^\ell}\frac{\langle \psi, \mc H^\ell_\zeta\psi\rangle_{L^2}}{ \|\psi\|^2_{H^1}}\geq \bar g_1\,.
\end{equation}
Since $J_{\ell,\zeta}=J_{\ell+\zeta,0}$ it is enough to show that
\begin{equation}
\label{assurdo!}
\liminf_{\ell\to+\infty}\inf_{\psi\perp \Psi_0^\ell}\frac{\langle \psi, \mc H^\ell_0\psi\rangle_{L^2}}{ \|\psi\|^2_{H^1}}>0\,.
\end{equation}
We argue by contradiction. If \eqref{assurdo!} does not hold there exists a sequence $\psi_\ell$ with 
$\|\psi_\ell\|_{H^1}=1$ and $\psi_\ell\perp\Psi^\ell_0$ such that 
\begin{equation*}
\langle \psi_\ell, \mc H^\ell_0\psi_\ell\rangle_{L^2}=\int_{-\ell}^{+\infty}\!  \big(|\psi'_\ell|^2 + V''(\bam)\psi_\ell^2\big)\, \rmd y\to 0\,.
\end{equation*}
By \eqref{gap1} we necessarily have $\psi_\ell\to 0$ in $L^2$. In view of  the boundedness of $V''(\bam)$ the formula  above gives the required contradiction.

By writing $\varphi=\langle\varphi,\Psi_\zeta^\ell\rangle_{L^2} \Psi_\zeta^\ell +\varphi^\perp$, from \eqref{gap2} and Young's inequality we have, for each $\gamma>0$,
\begin{equation}
\begin{split}\label{gamma}
\langle \varphi, \mc H^\ell_\zeta\varphi\rangle_{L^2}&= \langle\varphi,\Psi_\zeta^\ell\rangle_{L^2} ^2\lambda_\zeta^\ell +\langle \varphi^\perp, \mc H^\ell_\zeta\varphi^\perp\rangle_{L^2}\\
&\geq \bar g_1 \left(\|\varphi\|^2_{H^1} -2 \langle\varphi,\Psi_\zeta^\ell\rangle_{L^2} \langle\varphi^\perp,\Psi_\zeta^\ell\rangle_{H^1}-\langle\varphi,\Psi_\zeta^\ell\rangle^2_{L^2} \|\Psi_\zeta^\ell\|^2_{H^1}  \right)
\\
&\geq \bar g_1 \left(\|\varphi\|^2_{H^1} -\gamma \langle\varphi^\perp,\Psi_\zeta^\ell\rangle_{H^1}^2-\langle\varphi,\Psi_\zeta^\ell\rangle^2_{L^2} (\|\Psi_\zeta^\ell\|^2_{H^1}+\gamma^{-1})  \right)\,.
\end{split}
\end{equation}
Since $\mc H^\ell_z \Psi_\zeta^\ell = \lambda^\ell_z \Psi_\zeta^\ell$, we easily deduce that $\|\Psi_\zeta^\ell\|^2_{H^1}$ is bounded uniformly in $\ell$ and $\zeta$. Moreover, by Schwarz's inequality and the orthogonality between $\varphi^\perp$ and $\Psi_\zeta^\ell$, choosing $\gamma$ small enough in \eqref{gamma}, we obtain 
\begin{equation}\label{gammapiccolo}
\langle \varphi, \mc H^\ell_\zeta\varphi\rangle_{L^2}\geq C  \left(\|\varphi\|^2_{H^1} -\langle\varphi,\Psi_\zeta^\ell\rangle^2_{L^2}   \right)\,.
\end{equation}
Using \eqref{ort} we get
\begin{equation*}
\begin{split}
|\langle\varphi,\Psi_\zeta^\ell\rangle_{L^2}|&
=\Big|\Big\langle\varphi,\Psi_\zeta^\ell+\frac{\partial_z m^\ell_\zeta}{\|\partial_z m^\ell_\zeta\|_{L^2}}\Big\rangle_{L^2}\Big| 
\\
&\leq \|\varphi\|_{H^1}\Big(\Big\|\Psi_\zeta^\ell -\frac{\bam'_\zeta}{\|\bam'_\zeta\|_{L^2}}\Big\|_{L^2} + \Big\|\frac{\bam'_\zeta}{\|\bam'_\zeta\|_{L^2}}+\frac{\partial_z m^\ell_\zeta}{\|\partial_z m^\ell_\zeta\|_{L^2}}\Big\|_{L^2}\Big)\,.
\end{split}
\end{equation*}
In order to bound the right hand side, we first claim that
\begin{equation}
\label{falso}
\Big\|\Psi_\zeta^\ell -\frac{\bam'_\zeta}{\|\bam'_\zeta\|_{L^2}}\Big\|_{L^2}\leq C \rme^{-2(\ell+\zeta)}\,.
\end{equation}
Indeed, a slightly weaker estimate is stated in \cite[Theorem~3.2]{bbb}. However, it is straightforward to modify the argument of the proof to get \eqref{falso}, see in particular \cite[page 336]{bbb}. The second term in the right hand side can be easily estimated using \eqref{mellezeta} which gives, taking into account that $m_\zeta^\ell(y)=\bam_\zeta(y)$ if $y>\zeta$ and that $\ell+\zeta>\kappa_1-2\delta_1$,
\begin{equation*}
\Big\|\frac{\bam'_\zeta}{\|\bam'_\zeta\|_{L^2}}+\frac{\partial_z m^\ell_\zeta}{\|\partial_z m^\ell_\zeta\|_{L^2}}\Big\|_{L^2}\leq C\kappa_1\rme^{-2\kappa_1}\,.
\end{equation*}
In conclusion, choosing $\kappa_1$ large enough, the previous bounds, together with \eqref{gammapiccolo}, give \eqref{i2} which completes the proof of statement (ii).

\smallskip\noindent
{\it Proof of statement (iii).} We notice that in the proof of statement (ii) the estimates of the terms $I_\ell^1$ and $I_\ell^3$ do not require $\zeta$ being the center of $m$, while $I_\ell^2$ can be easily estimated from above by the $H^1$-norm of $m-m^\ell_\zeta$.
\end{proof}

\section{Compactness}
\label{sec:4}

We are now ready to analyze the two dimensional functional. In this section we prove the compactness statement in Theorem~\ref{teorema1}. Let us consider a sequence $u_\eps$ in $X$ such that $F_\eps(u_\eps)\leq C_3$, namely
\begin{equation}
  \label{comp1}
  \int_0^{+\infty}\!\int_{-\ell_\eps}^{+\infty}\! (\partial_xu_\eps)^2 \,\rmd y\, \rmd x
  + \int_0^{+\infty} \! \eps^{-2}\left[\mc
    F_{\ell_\eps}(u_\eps(x,\cdot)) -\frac83\right]\, \rmd x\leq C_3\,, 
\end{equation}
where $\ell_\eps-\frac12 \log\eps^{-1}\to \alpha$, and $u_\eps$ satisfies the boundary conditions \eqref{bc} for some $w_\eps$ such that \eqref{datobordo} holds.
\begin{remark}
By Schwarz's inequality and the bound \eqref{comp1}, for any $x_1,x_2$ in $[0,+\infty)$,
\begin{equation}
\label{comp2}
\|u_\eps(x_1,\cdot)-u_\eps(x_2,\cdot)\|^2_{L^2}\leq C_3 |x_1-x_2|\,.
\end{equation}
\end{remark}
Given a sequence $M_\eps\to +\infty$ such that $M_\eps\eps^2\to 0$ as $\eps\to 0$, we define the set of \emph{good} points in $(0,+\infty)$ as
\begin{equation}\label{buoni}
  B_\eps=\Big\{ x\in (0,+\infty)\,:\ \mc
    F_{\ell_\eps}(u_\eps(x,\cdot))-\frac83\leq M_\eps\eps^2\Big\}\,. 
\end{equation}
The bound \eqref{comp1} yields $|B_\eps^\complement|\leq C_3/M_\eps$ (here $|B|$ is the Lebesgue measure of the Borel set $B\subset \bb R$). Moreover, since the bound \eqref{comp2} guarantees that the map $x\mapsto u_\eps(x,\cdot)$ is continuous from $(0,+\infty)$ to $\mc X$ (see the previous section for the definition of $\mc X$), the lower semicontinuity of $\mc F_\ell$ on $\mc X$ implies that the map $x\mapsto \mc F_{\ell_\eps}(u_\eps(x,\cdot))$ is lower semicontinuous and hence the set $B_\eps$ is closed.

We now show how to construct the sequence $\phi_\eps$. Recalling the assumption \eqref{datobordo} on the boundary datum, Theorem~\ref{lemma0} implies that if $\eps$ is small enough and $x\in B_\eps\cup\{0\}$ then there exists a unique center of $u_\eps(x,\cdot)$, that we denote by $\phi_\eps(x)$. Let us note that the function $\phi_\eps$ is measurable on $B_\eps$. This can be easily deduced by the continuity in the uniform topology of the map that to each function in the set $\mc T^\ell(\delta,\kappa)$ associates its center, see \cite[Proposition 3.2]{bdp}, and the measurability of the map $B_\eps\ni x\mapsto u(x,\cdot)$ with respect to the Borel $\sigma$-algebra associated to the uniform topology. 

Since $\ell_\eps-\frac12\log\eps^{-1}\to \alpha$, in view of \eqref{gap}, there exists a constant $C_4>0$, depending on $\alpha$, such that the following bounds hold,
\begin{equation}
\label{phi}
\begin{split}
  & \phi_\eps(x)\geq -\frac14\log(C_4M_\eps)
  \\
  &\|u_\eps(x,\cdot)-m^{\ell_\eps}(\cdot,\phi_\eps(x))
  \|^2_{H^1} 
  \leq C_0 M_\eps\eps^2
\end{split}
\qquad \forall\  x\in B_\eps
\end{equation}
and 
\begin{equation}
  \label{phi0}
  \lim_{\eps\to0}\phi_\eps(0)=0\,,
  \qquad
  \|u_\eps(0,\cdot)-m^{\ell_\eps}(\cdot,\phi_\eps(0))\|^2_{H^1}\leq  
  \eps \, \eta_\eps\,,
\end{equation}
where, in view of \eqref{datobordo},
\begin{equation}
\label{oe}
\eta_\eps:=\eps^{-1}\Big(\mc F_{\ell_\eps}(w_\eps)-\frac83\Big) \to 0\,.
\end{equation}
Since $B_\eps^\complement$ is a countable union of disjoint open intervals, we extend $\phi_\eps$ to a  function on $[0,+\infty)$ by defining it in each interval of $B_\eps^\complement$ as the affine interpolation of the values of $\phi_\eps$ at the endpoints. 

The compactness stated in Theorem~\ref{teorema1} is a consequence of the following two lemmata. Indeed, Lemma~\ref{holderphi} yields the precompactness of $\phi_\eps$ in the uniform topology, while Lemma~\ref{utilde} together with \eqref{stim1} imply \eqref{tilde0}.
\begin{lemma}
  \label{holderphi}
  Let $\phi_\eps$ be defined as above. Then there exists a positive constant $C_5$
  such that, for any  $x_1\,,x_2\in [0,+\infty)$,
  \begin{equation}\label{phi12}
    |\phi_\eps(x_1)-\phi_\eps(x_2)|\wedge|\phi_\eps(x_1)-\phi_\eps(x_2)|^2 
    \leq C_5 \Big(|x_1-x_2| +\frac1{M_\eps}+M_\eps\eps^2 +
      \eps\eta_\eps\Big)\,,
  \end{equation}
where $M_\eps$ is the sequence in \eqref{buoni} and $\eta_\eps$ is the sequence defined in \eqref{oe}.
\end{lemma}
\begin{proof}
Since $\phi_\eps$ is affine outside $B_\eps$ and $|B_\eps^\complement|\leq C_3/M_\eps$, it is enough to prove that there exists $C_6>0$ such that for any $x_1,x_2\in B_\eps\cup\{0\}$,
\begin{equation}
\label{phi13}
    |\phi_\eps(x_1)-\phi_\eps(x_2)|\wedge|\phi_\eps(x_1)-\phi_\eps(x_2)|^2 
    \leq C_6 \big(|x_1-x_2| +M_\eps\eps^2 +
      \eps\eta_\eps\big)\, .
  \end{equation}
If $x_1,x_2\in B_\eps\cup\{0\}$ the bound \eqref{lip} implies 
  \begin{equation}
    \begin{split}
      &|\phi_\eps(x_1)-\phi_\eps(x_2)|\wedge|\phi_\eps(x_1)-\phi_\eps(x_2)|^2\\
      &\qquad\leq A
      \|m^{\ell_\eps}(\cdot,\phi_\eps(x_1))- 
      m^{\ell_\eps}(\cdot,\phi_\eps(x_2))\|^2_{L^2}\\  
      &\qquad \leq 2A\big(
        \|u_\eps(x_1,\cdot)-u_\eps(x_2,\cdot)\|^2_{L^2}+\|\widetilde
        u_\eps(x_1,\cdot)-\widetilde
        u_\eps(x_2,\cdot)\|^2_{L^2}\big),
\end{split}
\end{equation}
where $\widetilde u_\eps(x,y):= u_\eps(x,y)- m^{\ell_\eps}(y,\phi_\eps(x))$.  By using \eqref{comp2}, \eqref{phi}, and \eqref{phi0} the bound \eqref{phi13} follows.
\end{proof}
Recall that $m^\ell_z(\cdot)\equiv m^\ell(\cdot,z)$ is defined in \eqref{nome}.
\begin{lemma}
\label{utilde}
Let $u_\eps$ be a sequence satisfying the bound \eqref{comp1}, let $\phi_\eps$ be defined as above, and set $\widetilde u_\eps(x,y):= u_\eps(x,y)-m^{\ell_\eps}(y,\phi_\eps(x))$, $(x,y)\in [0,+\infty)\times \bb R$. For each $R>0$ the sequence $\widetilde u_\eps$ converges to $0$, as $\eps\to 0$, in $L^2((0,R)\times\bb R)$.
\end{lemma}
\begin{proof}
The estimate \eqref{phi} trivially implies that, for any $R>0$,
\begin{equation}
\label{tildebuoni}
\lim_{\eps\to 0}\int_{[0,R)\cap B_\eps}\!\int_{-\infty}^{+\infty}\!  |\widetilde u_\eps|^2 \,\rmd y\, \rmd x=0\,.
\end{equation}
By \eqref{comp2}, \eqref{lip} and Lemma~\ref{holderphi}, for a suitable constant $C_7>0$ we get, for any $x_1,x_2$ in $[0,+\infty)$, 
\begin{equation*}
\|\widetilde u_\eps(x_1,\cdot)-\widetilde u_\eps(x_2,\cdot)\|^2_{L^2}\leq C_7 \Big(|x_1-x_2| +\frac1{M_\eps}+M_\eps\eps^2 + \eps \eta_\eps\Big)\,.
\end{equation*}
This, together with \eqref{tildebuoni} and the fact that $|B_\eps^\complement|\leq C_3/M_\eps$, concludes the proof. 
\end{proof}

\section{$\Gamma$-convergence}
\label{sec:5}

In this section we conclude the proof of the main result by proving the $\Gamma$-convergence of the functionals $F_\eps$. 

\begin{proof}[Proof of Theorem~\ref{teorema1}: $\Gamma$-liminf.] The formal statement of the $\Gamma$-liminf inequality is the following. For each $u\in X$ and each sequence $u_\eps$ converging to $u$ in $X$, it holds $\liminf_{\eps\to 0} F_\eps(u_\eps)\leq F^\alpha(u)$. In view of the compactness result, the $\Gamma$-liminf is achieved once we show that, for each $\phi\in C([0,+\infty))$ and each sequence $u_\eps$ converging to $\bam_\phi$ in $X$,
\begin{equation}\label{lower1}
\liminf_{\eps\to0} F_\eps(u_\eps)\geq \mc G^\alpha(\phi)\,.
\end{equation}

Fix $\phi\in C([0,+\infty))$ and a sequence $u_\eps$ converging to $\bam_\phi$. Without loss of generality we can assume that $F_\eps(u_\eps)\leq C_8$. Therefore, in view of Lemmata~\ref{holderphi} and \ref{utilde}, by extracting, if necessary, a subsequence, there exists a  sequence $\phi_\eps$ converging to $\phi$ in $C([0,+\infty))$ such that 
$u_\eps=m^{\ell_\eps}(\cdot,\phi_\eps)+\widetilde u_\eps$, with $\widetilde u_\eps$ converging to zero in $L^2((0,R)\times\bb R)$, for any $R>0$. Let $B_\eps$ be the set of good points as defined in \eqref{buoni}. Then
\begin{equation}
\label{lowerbound}
\begin{split}
&F_\eps(u_\eps)\geq \int_0^{+\infty}\! \int_{-\infty}^{+\infty}\! (\partial_x u_\eps)^2\,\rmd y\, \rmd x +\!\!\int_{B_\eps}\!\!\!\!\eps^{-2}\Big[\mc F_{\ell_\eps}(u_\eps(x,\cdot))-\frac83\Big] \, \rmd x \\
&\quad= \int_0^{+\infty}\! \int_{-\infty}^{+\infty}\! (\partial_x u_\eps)^2 \,\rmd y\, \rmd  x+\!\!\int_{B_\eps}\!\!\!\!\!\eps^{-2}\Big[\mc F_{\ell_\eps}(m^{\ell_\eps}(\cdot,\phi_\eps(x)))-\frac83\Big] \, \rmd x +\mc R_\eps\,.
\end{split}
\end{equation}
Since, for each $R>0$, $u_\eps\to \bam_\phi$ in $L^2((0,R)\times\bb R)$, the lower semicontinuity of the map $u\mapsto \|\partial_x u\|^2_{L^2((0,R)\times \bb R)}$ with respect to the $L^2$-convergence gives 
\begin{equation}\label{step1}
\begin{split}
& \liminf_{\eps\to 0}\int_0^R\! \int_{-\infty}^{+\infty}\! (\partial_x u_\eps(x,y))^2 \,\rmd x \,\rmd y \\ & \qquad \geq \int_0^R\!\int_{-\infty}^{+\infty}\! (\partial_x \bam_{\phi(x)}(y))^2\,\rmd y\, \rmd  x=\frac43\int_0^R\! \phi'(x)^2 \,\rmd x \,.
\end{split}
\end{equation}
The estimate of the second term on the right hand side of \eqref{lowerbound} is a direct consequence of Lemma~\ref{fucking}. Indeed, since $\epsilon^{-2} e^{-4\ell_\epsilon} \to e^{-4\alpha}$, by Fatou's Lemma and the fact that $|B_\eps^\complement|\to 0$, we get
\begin{equation}\label{step2}
\liminf_{\eps\to 0}\int\limits_{B_\eps\cap
  (0,R)}\!\eps^{-2}\Big[\mc
  F_{\ell_\eps}(m^{\ell_\eps}(\cdot,\phi_\eps(x)))-\frac83\Big] \, \rmd x\geq 
16\,e^{-4\alpha}\, 
\int_0^R\!\rme^{-4\phi(x)}\,\rmd x\,, 
\end{equation}
for any $R>0$.

Finally, we need to estimate  $\mc R_\eps$ as defined in \eqref{lowerbound}, i.e.,
\begin{equation*}
\mc R_\eps= \!\!\int_{B_\eps}\!\!\!\!\!\eps^{-2}\left[\mc
  F_{\ell_\eps}(u_\eps(x,\cdot))-\mc F_{\ell_\eps}(m^{\ell_\eps}(\cdot,\phi_\eps(x)))\right] \, \rmd x\,.
\end{equation*}
By \eqref{phi}, for any $x\in B_\eps$,
\begin{equation*}
\rme^{-4\phi_\eps(x)}
\|\widetilde u_\eps(x,\cdot)\|_{H^1}\leq C_4C_0^{\frac12} M_\eps^\frac32\eps\,.
\end{equation*}
Thus, if we further choose $M_\eps$ such that $M_\eps^3\eps^2\to 0$ as $\eps\to 0$, from \eqref{gapimprove} we get
\begin{equation}
\label{reps}
\liminf_{\eps\to 0}\mc R_\eps\geq \frac{1}{C_1}
\lim_{\eps\to0}\int_{B_\eps}\!\eps^{-2}\|\widetilde u_\eps(x,\cdot)\|_{H^1}^2\, \rmd x\, .
\end{equation}
The bound \eqref{lower1} follows by \eqref{lowerbound}, \eqref{step1}, \eqref{step2}, and \eqref{reps}.
\end{proof}
We note that the previous arguments show that if the energy of the sequence $u_\eps$ converges to $\mc G^\alpha(\phi)$ then $u_\eps(x,\cdot)$, $x\in B_\eps$, is actually close in $H^1(\bb R)$ topology to the ``right'' one dimensional profile, with an explicit control on the norm. The precise statement is given in the following remark.
\begin{remark}
\label{caporetto}
Take a sequence $u_\eps$ with $F_\eps(u_\eps)\le C_9$ and decompose $u_\eps$ as $u_\eps=m^{\ell_\eps}(\cdot,\phi_\eps)+\widetilde u_\eps$, where $\phi_\eps$ is the sequence constructed in Section~\ref{sec:4}. If $u_\eps\to\bam_\phi$, for some $\phi\in C([0,+\infty))$, and satisfies
\begin{equation*}
\liminf_{\eps\to0} F_\eps(u_\eps)=\mc G^\alpha(\phi) < +\infty\,, 
\end{equation*}
then from \eqref{reps} we easily deduce that
\begin{equation*}  
\lim_{\eps\to0} \int_{B_\epsilon}
  \!\eps^{-2}\|\widetilde u_\eps(x,\cdot)\|_{H^1}^2\,\rmd x=0\,.
\end{equation*}
\end{remark}
\begin{proof}[Proof of Theorem~\ref{teorema1}: $\Gamma$-limsup]
We now show that for any function $u\in X$ of the form $u=\bam_\phi$, with $\phi\in C([0,+\infty))$, we can construct a sequence $\bar u_\eps$ such that
\begin{equation}
\label{optimal}
\liminf_{\eps\to0} F_\eps(\bar u_\eps)=\mc G^\alpha(\phi)\,.
\end{equation}
We observe that for each $\phi\in C([0,+\infty))$ such that $\phi(0)=0$ and $\mc G^\alpha(\phi)<+\infty$ there exists a sequence $\phi_n$, with $\text{supp}\, \phi_n\subset (n^{-1},+\infty)$ and $\phi_n\ge -n$, converging to $\phi$ and satisfying $\lim_n\mc G^\alpha(\phi_n) =\mc G^\alpha(\phi) $. By standard properties of the $\Gamma$-limsup, see e.g.~\cite[Remark~1.29]{marito1}, it is therefore enough to construct the recovery sequence for $\phi\in C([0,+\infty))$ bounded from below and with $\text{supp}\, \phi\subset (\delta,+\infty)$, $\delta>0$.

Let $\zeta_\epsilon$ be a center of the boundary condition $w_\epsilon$. In view of \eqref{datobordo} and Theorem~\ref{lemma0}, $\zeta_\epsilon$ is in fact the unique center of $w_\epsilon$. Moreover, by \eqref{stim1}, the real sequence $\zeta_\epsilon$ converges to zero as $\epsilon\to 0$. By redefining $\ell_\epsilon$ we can thus assume, and do it now, that $\zeta_\epsilon=0$.

We claim that the following sequence does the job,
\begin{equation}
\label{recovery}
\bar u_\eps(x,y):=
\begin{cases}
m^{\ell_\eps}(y,\phi(x)) & 
\hbox{ if }  
(x,y)\in [\epsilon,+\infty)\times \bb R \,,
\\
m^{\ell_\eps}\big(y,0 \big) 
+ \frac{\epsilon-x}{\epsilon} \widetilde w_\eps(y) 
& 
\hbox{ if }  
(x,y)\in [0,\epsilon)\times \bb R \,,
\end{cases}
\end{equation}
where $\widetilde w_\eps:=w_\eps-m^{\ell_\eps}_0$.  

In the sequel we  use the notation $F_\eps(\cdot,A)$ for the localization of the functional $F_\eps$ on the set $A\subset (0,+\infty)\times \bb R$. Since $m^\ell_z=m_0^{\ell+z}$, by Lemma~\ref{fucking} it follows that for each $\bar z\in \bb R$ there exists $\bar\ell>0$ such that
\begin{equation*}
\rme^{-4\ell}\left[\mc F_\ell(m^\ell_z)-\frac83\right]\le 17\,\rme^{-4z}\qquad \forall\, z\in [\bar z,+\infty) \quad \forall\, \ell>\bar\ell\,.
\end{equation*}
Since we assumed $\phi$ to be bounded from below, by Lemma~\ref{fucking}, dominated convergence, and \eqref{dz}, we deduce
\begin{equation}
\label{dadelta}
\lim_{\eps\to 0} F_\eps( \bar u_\eps ,(\delta,
+\infty)\times(-\ell_\epsilon,+\infty))
= \int_{\delta}^{+\infty}\! \Big[ \frac43 \phi'(x)^2 
  + 16\,e^{-4\alpha}  \,\rme^{-4 \phi(x)} \Big] \,\rmd x \,.
\end{equation}
We now show that
\begin{equation}
\label{raccordo}
\lim_{\eps\to 0} F_\eps(\bar u_\eps,
(0,\delta)\times(-\ell_\epsilon,+\infty))= 16 \, \rme^{-4\alpha}\,\delta\,.
\end{equation}
Since $\text{supp}\, \phi\subset (\delta,+\infty)$, \eqref{optimal} is a straightforward consequence of \eqref{dadelta} and \eqref{raccordo}.
 
To conclude, we are left with the proof of \eqref{raccordo}. As it follows from \eqref{datobordo} and \eqref{gapimprove}, $\lim_{\epsilon\to 0} \eps^{-1}\|\widetilde w_\eps\|^2_{H^1}=0$ and therefore
\begin{equation*}
\lim_{\eps\to 0}\int_0^\delta
\!\int_{-\ell_\epsilon}^{+\infty}\!|\partial _x \bar u_\eps|^2\,\rmd y\, \rmd x=
\lim_{\eps\to 0}
\int_0^\epsilon\!\int_{-\ell_\epsilon}^{+\infty}\! \epsilon^{-2}
|\widetilde w_\epsilon(y)|^2\,\rmd y\, \rmd x=0\,.
\end{equation*}
On the other hand, since $\phi(x) =0 $ for $x\in [0,\delta]$, $\bar u_\eps(x,\cdot) = m^{\ell_\eps}_0(\cdot)$ in $(\eps,\delta)$, then
\begin{equation*}
\begin{split}
  & \int_0^\delta \!\eps^{-2}\Big[\mc
    F_\eps(\bar u_\eps(x,\cdot) )-\frac83\Big]\, \rmd x
  \\ & \qquad =  \int_0^\delta\!\eps^{-2}\Big[ 
    \mc F_\eps(m^{\ell_\eps}_0) -\frac83 \Big]\, \rmd x + \int_0^\epsilon\!\eps^{-2}\Big[\mc
    F_\eps(\bar u_\eps(x,\cdot)  )- \mc F_\eps(m^{\ell_\eps}_0) \Big]\, \rmd x\,.
\end{split}
\end{equation*}
As $\ell_\epsilon-\frac12\log\epsilon^{-1}\to \alpha$, by Lemma~\ref{fucking},
\begin{equation*}
\limsup_{\epsilon\to 0} \:
  \int_0^\delta\!\eps^{-2}\Big[ 
    \mc F_\eps(m^{\ell_\eps}_0) -\frac83 \Big]\, \rmd x
  =16 \, \rme^{-4\alpha}\,\delta\,.
\end{equation*}
As noted before $\lim_{\epsilon\to 0} \eps^{-1}\|\widetilde w_\eps\|^2_{H^1}=0$; therefore by  Theorem~\ref{lemma0}, item (iii),
\begin{equation*}
  \lim_{\epsilon\to 0} \int_0^\epsilon\!\eps^{-2}\left[\mc
    F_\eps(\bar u_\eps(x,\cdot))- \mc F_\eps(m^{\ell_\eps}_0) \right]\, \rmd x
  =0,
\end{equation*}
which completes the proof of \eqref{raccordo}. 
\end{proof}

\appendix
\section{Sharp estimates on the constrained minimizer} 
\label{sec:a}

In this appendix we prove the sharp estimates concerning $m_z^\ell$ and its convergence to $\bam_z$. We regard the boundary value problem \eqref{eulero} as a one dimensional Newtonian system with potential $-V$ and mass equal to two. Accordingly, the space variable $y$ is interpreted as the time and denoted by $t$. 
 
\begin{proof}[Proof of Proposition~\protect\ref{prop_ml}]
Given $T>0$, we denote by $m_T(t)$, $t\in [-T,0]$, the solution to the boundary value problem
\begin{equation}
\label{eulero1}
\begin{cases}
-2m'' +V'(m)=0 & \hbox{in}\ (-T,0)\,,\\
m(-T)=-1\,,\ m(0)=0\,.
\end{cases}
\end{equation}
Integrating \eqref{eulero1} by using the conservation of the Newtonian energy, we get that $m_T(t)$ is the strictly increasing function on $[-T,0]$ such that
\begin{equation}
\label{mt}
-t = \int_{m_T(t)}^0\!\frac{\rmd a}{\sqrt{V(a)+E_T}}\qquad\forall\, t\in [-T,0]\,,
\end{equation}
where $E_T$ is implicitly defined by the condition,
\begin{equation}
\label{et}
T = \int_{-1}^0\!\frac{\rmd a}{\sqrt{V(a)+E_T}}\,.
\end{equation}

In the sequel we denote by $C$ a strictly positive constant, independent of $T$, whose numerical value may change from line to line. By \cite[Lemma~A.1]{bbbeq}, 
\begin{equation}
\label{eti}
\lim_{T\to\infty} \rme^{4T}\, E_T = 64
\end{equation} 
and 
\begin{equation}
\label{mtb}
\sup_{t\in (-T,0)} \big| m_T(t) - \bam(t)\big| \le C \rme^{-2T}\qquad\forall\,T\ge 1\,.
\end{equation} 
We now observe that, for any $y\in (-\ell,z)$, 
\begin{equation}
\label{mtt}
\begin{split}
m^\ell_z(y) & = m_{\ell+z}(y-z)\,, \\ 
\partial_z m^\ell_z(y) & = \partial_T m_{\ell+z}(y-z) - m_{\ell+z}'(y-z)\,, \\ \partial_{zz} m^\ell_z(y) & = \partial_{TT} m_{\ell+z}(y-z) - 2\partial_T m_{\ell+z}'(y-z) + m_{\ell+z}''(y-z)\,.
\end{split}
\end{equation}
The bound \eqref{stim1} follows by \eqref{mtb}. We next show  
\begin{eqnarray}
\label{mtb1} && \!\!\!\!\!\!\!\!\!\!\!\!\sup_{t\in (-T,0)} \big| m_T'(t) - \bam'(t)\big| \le C \rme^{-2T}\qquad\forall\,T\ge 1\,, \\ \label{mtb2} && \!\!\!\!\!\!\!\!\!\!\!\!\sup_{t\in (-T,0)} \big| m_T''(t) - \bam''(t)\big| \le C \rme^{-2T}\qquad\forall\,T\ge 1\,, \\ \label{mtb3} && \!\!\!\!\!\!\!\!\!\!\!\!\sup_{t\in (-T,0)} \left\{\big|\partial_Tm_T(t)\big| +\big|\partial_Tm_T'(t)\big| + \big|\partial_{TT}m_T(t)\big| \right\} \le C\rme^{-2T}\qquad\forall\,T\ge 1\,, \qquad
\end{eqnarray}
which imply the estimates \eqref{mellezeta} and \eqref{mellezeta2}.

\medskip
\noindent{\it Proof of (\ref{mtb1})}. Since $\bam'(t) = \sqrt{V(\bam(t))}$ , $m_T'(t) = \sqrt{V(m_T(t)) +E_T}$, and $\bam(0)=m_T(0)=0$, we have 
\begin{equation*}
-1<m_T(t)<\bam(t)<0 \qquad\forall\, t\in(-T,0)\,.
\end{equation*}
Hence, for any $t\in(-T,0)$,
\begin{equation}
\label{mt1}
\big|m_T'(t) - \bam'(t)\big| \le \sqrt{E_T} + \frac{V(\bam(t)) - V(m_T(t))}{\sqrt{V(\bam(t))}} \le \sqrt{E_T} + 4 \big|m_T(t) - \bam(t)\big|\,,
\end{equation}
where, in the last inequality, we used that, by the explicit expression \eqref{p1} of $V$, $V(b)-V(a) \le 4\sqrt{V(b)}(b-a)$ for $-1<a<b<0$. The bound \eqref{mtb1} now follows by \eqref{eti}, \eqref{mtb}, and \eqref{mt1}.

\medskip
\noindent{\it Proof of (\ref{mtb2})}. Since $m_T''(t) - \bam''(t) = V'(m_T(t)) - V'(\bam(t))$ and $|V''(a)|\le 16$ for $-1<a<0$, the bound \eqref{mtb2} is an immediate consequence of \eqref{mtb}.

\medskip
\noindent{\it Proof of (\ref{mtb3})}. Taking the derivatives with respect to the variable $T$ in the identities \eqref{mt}, \eqref{et} and $m_T'(t) = \sqrt{V(m_T(t)) +E_T}$, we compute,
\begin{equation}
\label{com}
\begin{split}
E_T' & := \frac{\rmd E_T}{\rmd T} = - 2\left[\int_{-1}^0\!\frac{\rmd a}{\big[V(a)+E_T\big]^{3/2}}\right]^{-1}\,, \\ E_T'' & := \frac{\rmd^2 E_T}{\rmd T^2} = -\frac 34 (E_T')^3 \int_{-1}^0\!\frac{\rmd a}{\big[V(a)+E_T\big]^{5/2}} \,,  \\ \partial_T m_T(t) & = - \frac{E_T'}2 \sqrt{V(m_T(t)) +E_T} \int_{m_T(t)}^0\!\frac{\rmd a}{\big[V(a)+E_T\big]^{3/2}}\,, \\ \partial_{TT} m_T(t) & = \left[\frac 12 \frac{V'(m_T(t))\, \partial_T m_T(t)+ 2E_T'}{V(m_T(t)) +E_T} + \frac{E_T''}{E_T'} \right] \partial_T m_T(t) \, \\ & \phantom{=} + \frac34 (E_T')^2  \sqrt{V(m_T(t)) +E_T}\int_{m_T(t)}^0\!\frac{\rmd a}{\big[V(a)+E_T\big]^{5/2}} \,, \\ \partial_{T} m_T'(t) & = \frac 12 \frac{V'(m_T(t))\, \partial_T m_T(t)+ E_T'}{\sqrt{V(m_T(t)) +E_T}}\,.
\end{split}
\end{equation}
By the change of variable $b=1+a$ it is straightforward to check that, for any integer $n\ge1$,
\begin{equation*}
\int_0^1\!\frac{\rmd b}{\big[4b^2+E_T\big]^{n/2}} \le
\int_{-1}^0\!\frac{\rmd a}{\big[V(a)+E_T\big]^{n/2}} \le
\int_0^1\!\frac{\rmd b}{\big[b^2+E_T\big]^{n/2}}\,.
\end{equation*}
Therefore,
\begin{equation}
\label{c2}
\frac{g_n(E_T)}{C} \le \int_{-1}^0\!\frac{\rmd a}{\big[V(a)+E_T\big]^{n/2}} \le C \, g_n(E_T) \qquad \forall\, T\ge 1\,,
\end{equation}
where
\begin{equation*}
g_n(E_T) = \left\{\begin{array}{ll} |\log E_T|\,, & \textrm{ if }n=1\,, \\ E_T^{(1-n)/2} & \textrm{ if }n=3,5\,. \end{array}\right.
\end{equation*}
By \eqref{c2} and \eqref{com} it follows that 
\begin{equation}
\label{c3}
\frac{E_T}{C} \le |E_T'| \le C\, E_T, \qquad |E_T''| \le C\, E_T \qquad \forall\, T\ge 1\,.
\end{equation}
Since $V(a)\ge V(m_T(t))$ for $a\in [m_T(t)),0]$, using \eqref{c2} and \eqref{c3}, from \eqref{com}  we get
\begin{equation}
\label{a5}
\big|\partial_T m_T(t) \big| \le  \frac{|E_T'|}2 \int_{m_T(t)}^0\!\frac{\rmd a}{\big[V(a)+E_T\big]^{1/2}} \le C\, E_T |\log E_T|\,. 
\end{equation}
Analogously, using also the explicit form \eqref{p1} of $V$, 
\begin{equation}
\label{a6}
\big|\partial_{TT} m_T(t) \big| \le \left[\frac{2\big|\partial_T m_T(t)\big|+ |E_T'|}{E_T} + \frac{|E_T''|}{|E_T'|}\right] \big|\partial_T m_T(t)\big| + C\frac{|E_T'|^2}{E_T} \le C E_T\,.
\end{equation}
Finally,
\begin{equation}
\label{a7}
\big|\partial_T m_T'(t) \big| \le \frac 12 \frac{4\big|\partial_T m_T(t)\big|+ |E_T'|}{\sqrt{E_T}}  \le C \,\sqrt{E_T}\,.
\end{equation}
The bound \eqref{mtb3} now follows by \eqref{eti}, \eqref{a5}, \eqref{a6}, and \eqref{a7}.

\medskip
\noindent{\it Proof of (\ref{salto})}. Recall that $m^\ell_z(y) = \bam_z(y)$ for $y\ge z$, whence $(m^\ell_z)'(y) + \partial_z m^\ell_z(y) = 0$ for $y> z$. Therefore, by \eqref{mtt} and \eqref{com},
\begin{equation*}
\left[(m^\ell_z)'\right](z) +  \left[\partial_z m^\ell_z\right](z) = 
\lim_{y\uparrow z} \left\{(m^\ell_z)'(y) + \partial_z m^\ell_z(y) \right\} = \lim_{y\uparrow z}\partial_T m_{T}(y-z){\Big|_{T=\ell+z}}= 0\,.
\end{equation*}
On the other hand, since $\bam'(0)=1$, 
\begin{equation*}
\left|\left[(m^\ell_z)'\right](z)\right| = \left| 1 - \sqrt{V(m_{\ell+z}(0))+E_{\ell+z}} \right| = \left| 1 - \sqrt{1+E_{\ell+z}}\right| \le E_{\ell+z}\,.
\end{equation*}
By \eqref{eti}, this concludes the proof of \eqref{salto}.

\medskip
\noindent{\it Proof of (\ref{lip})}. Without loss of generality we assume $z_1<z_2$. Since
\begin{equation*}
m^\ell_{z_2}(y) < 0 < m^\ell_{z_1}(y) = \bam_{z_1}(y)\qquad\forall\, y\in (z_1,z_2)
\end{equation*}
then
\begin{equation}\label{dalbasso}
\begin{split}
\|m^\ell_{z_1}-m^\ell_{z_2}\|^2_{L^2} &\ge \int_{z_2}^{+\infty}\! |\bam_{z_1}(y)-\bam_{z_2}(y)|^2\,\rmd y\\
&= \int_{0}^{+\infty}\! |\bam(y+z)-\bam(y)|^2\,\rmd y=: G(z)\,, 
\end{split}
\end{equation}
with $z=z_2-z_1$. By differentiating, 
\begin{equation*}
G'(z)=2\int_0^{+\infty}\!\bam'(y+z)(\bam(y+z)-\bam(y))\,\rmd y\,,
\end{equation*}
whence $G'(0)=0$ and, since $\bam$ is strictly increasing, $G'$ is strictly increasing. Moreover,
 \begin{equation*}
G''(0)=2\int_0^{+\infty}\!|\bam'(y)|^2\,\rmd y=\frac43\,.
\end{equation*}
The above properties of the function $G$ imply $G(z) \ge C\, z^2 \wedge z$ for any $z\ge 0$. In view of \eqref{dalbasso}, this yields the lower bound of the estimate \eqref{lip}.

To prove the upper bound we analyze separately the cases $z_2-z_1\le 1$ and $z_2-z_1> 1$. In the first case, we use the Schwarz's inequality and \eqref{mellezeta2} to write
\begin{equation*}
\begin{split}
\|m^\ell_{z_1}-m^\ell_{z_2}\|^2_{L^2} & = \left\|\int_{z_1}^{z_2}\!\partial_z m^\ell_z \,\rmd z \right\|^2_{L^2} \le (z_2-z_1) \int_{z_1}^{z_2}\!\left\|\partial_z m^\ell_z \right\|^2_{L^2}\,\rmd z \\ & \le \left(2\|\bam'\|^2_{L^2} + 2A^2(\ell+z_2) \rme^{-4(\ell+z_1)}\right) (z_2-z_1)^2\,.
\end{split}
\end{equation*}
In the second case, recalling the definition of $m^\ell_{z_i}$ outside $(-\ell,z_i)$ and using  \eqref{stim1}, we have
\begin{equation*}
\begin{split}
\|m^\ell_{z_1}-m^\ell_{z_2}\|^2_{L^2} & \le \int_{-\ell}^{+\infty}\!2\big[\bam_{z_1}(y)-\bam_{z_2}(y)\big]^2 \,\rmd y + 8A^2(\ell+z_2) \rme^{-4(\ell+z_1)} \\ & \le C\, (z_2-z_1)\,.
\end{split}
\end{equation*}
The proposition is thus proved.
\end{proof}

\end{document}